\documentclass[review,hidelinks,onefignum,onetabnum]{siamart220329}

\usepackage{lipsum}
\usepackage{amsmath,amssymb,amsfonts}
\usepackage{graphicx}
\usepackage{epstopdf}
\usepackage{algorithmic}
\ifpdf
  \DeclareGraphicsExtensions{.eps,.pdf,.png,.jpg}
\else
  \DeclareGraphicsExtensions{.eps}
\fi

\usepackage{tabularx}
\usepackage{cleveref}
\usepackage{hyperref}
\usepackage{tikz}
\usepackage{caption}
\usepackage{subcaption}
\usepackage{stmaryrd}
\usepackage{color}
\usepackage{xcolor}%
\usepackage{booktabs}

\usetikzlibrary{shapes,arrows,chains,decorations.pathreplacing}

\newcommand{\abs}[1]{\left\vert#1\right\vert}

\DeclareMathOperator{\op}{op}

\ifx\ispreprint\undefined
\newsiamremark{remark}{Remark}
\newsiamremark{hypothesis}{Hypothesis}
\crefname{hypothesis}{Hypothesis}{Hypotheses}
\else
\usepackage{amsthm}
\theoremstyle{plain}
\newtheorem{lemma}{Lemma}
\newtheorem{proposition}{Proposition}
\newtheorem{theorem}{Theorem}
\theoremstyle{definition}
\newtheorem{definition}{Definition}
\theoremstyle{remark}
\newtheorem{remark}{Remark}

\usepackage{mathtools}
\fi

\headers{Martingales in Stochastic Rounding}{P. de Oliveira Castro, E.-M. El arar, E. Petit and D. Sohier}

\title{Error Analysis of Sum-Product Algorithms under Stochastic Rounding\thanks{Version of \today.
\funding{This work was funded by the HOLIGRAIL (ANR-23-PEIA-0010), the INTERFLOP (ANR-20-CE46-0009), and FPT-4 (ANR-24-CE46-7572) projects.}}}

\author{
    Pablo de Oliveira Castro\footnotemark[1]\thanks{Université Paris-Saclay, UVSQ, Li-PaRAD, Saint-Quentin en Yvelines, France 
    (\email{pablo.oliveira@uvsq.fr}, \email{devan.sohier@uvsq.fr}).}
    \and
    El-Mehdi El arar\footnotemark[2]\thanks{Université de Rennes, Inria, IRISA, Rennes, France 
    (\email{el-mehdi.el-arar@inria.fr}).}
    \and
    Eric Petit\footnotemark[3]\thanks{Intel Corp , Portland, USA
    (\email{eric.petit@intel.com}).}
    \and
    Devan Sohier\footnotemark[1]
}

\usepackage{amsopn}

\ifpdf
\hypersetup{
  pdftitle={Error Analysis of Sum-Product Algorithms under Stochastic Rounding},
  pdfauthor={P. de Oliveira Castro, E.-M. El Arar, E. Petit and D. Sohier}
}
\fi

\newcommand{\mido}[1]{\textcolor{blue}{#1}}
\renewcommand{\textcolor}[2]{#2}

\nolinenumbers

\begin{document}

\maketitle

\begin{abstract}
The quality of numerical computations can be measured through their forward error, for which finding good error bounds is challenging in general. For several algorithms and using stochastic rounding (SR), probabilistic analysis has been shown to be an effective alternative for obtaining tight error bounds. This analysis considers the distribution of errors and evaluates the algorithm's performance on average. Using martingales and the Azuma-Hoeffding inequality, it provides error bounds that are valid with a certain probability and in $\mathcal{O}(\sqrt{n}u)$ instead of deterministic worst-case bounds in $\mathcal{O}(nu)$, where $n$ is the number of operations and $u$ is the unit roundoff.
In this paper, we present a general method that automatically constructs a martingale for any computation scheme with multi-linear errors based on additions, subtractions, and multiplications. We apply this generalization to algorithms previously studied with SR, such as pairwise summation and the Horner algorithm, and prove equivalent results. We also analyze a previously unstudied algorithm, Karatsuba polynomial multiplication, which illustrates that the method can handle reused intermediate computations.
\end{abstract}



\section{Introduction}
\emph{Stochastic Rounding} (SR) is a rounding mode for floating-point numbers in which the rounding direction is chosen at random, inversely proportionally to the relative distance to the nearest representable values. SR is an alternative to the more common deterministic rounding that has drawn attention in recent years~\cite{muller2018handbook}, in particular due to its resilience to stagnation~\cite{xia2023influences, xia2023convergence}; the phenomenon in which the accumulator in long summations become so big that the remaining individual terms to be summed become negligible with respect to the precision in use, even if their exact sum is not. Indeed, for summations, \mido{round-to-nearest} (RN) has worst-case error bounds proportional to the number of floating-point operations $n$. With high probability, SR has error bounds~\cite{survey} proportional to $\sqrt{n}$.
SR is more robust than RN because the randomness removes the bias in the accumulation of errors.

For example, during parameter updates in deep learning, SR avoids stagnation, particularly when using low-precision formats for computations or storage~\cite{gupta}.  In \textit{gradient descent}, when computing the minimum of a function using RN-binary16 precision~\cite{xia2023influences, xia2023convergence}, it has been observed that the gradient can approach zero too quickly, causing the update to be lost due to limited precision. SR mitigates this issue by maintaining some accuracy on average, preventing stagnation in such scenarios.

Until now, computing SR probabilistic error bounds has been done case-by-case with algorithm-dependent proofs. Available proofs in the literature fall into two main schemes. A first proof scheme~\cite{theo21stocha,ilse, positive, hallman2022precision} models the algorithm's error as a stochastic process $\Psi_k$, shows that it is a martingale, and computes an error bound with the Azuma-Hoeffding concentration inequality. A second scheme~\cite{arar2022stochastic}, bounds the variance of $\Psi_k$ and applies \mido{the} Chebyshev concentration inequality.

This paper generalizes the computation of SR error bounds to all algorithms that can be modeled as a computation \mido{Directed Acyclic Graph} (DAG) comprised of sums, subtractions, and multiplications, as long as no multiplication node has two children sharing a common ancestor that is not an input (this is, in particular, true of all computation trees).

Section~\ref{sec:sum-product-trees} proves through structural induction that the errors in such computation DAGs form a martingale. It also gives a systematic recursive formulation to bound the martingale increments, allowing the use of Azuma-Hoeffding inequality to compute a probabilistic error bound of the whole computation DAG.

In Section~\ref{sec:examples}, we apply the method to generalize previous results on pairwise summation~\cite{hallman2022precision} and Horner's polynomial evaluation~\cite{arar2022stochastic}. Moreover, to the best of our knowledge, we are the first to investigate Karatsuba polynomial multiplication under SR. We demonstrate the applicability of the generalization proved in Section~\ref{sec:sum-product-trees} to bound the forward error of this algorithm under SR. 

With SR, the algorithmic errors are captured through a stochastic process. We
propose to use an important result from martingale theory, the Doob-Meyer
decomposition~\cite{doob1953stochastic}, to decompose the error stochastic process into a martingale and a predictable drift. We show
that the computation trees analyzed in previous sections always have a zero drift in
such decomposition. The paper closes with a discussion of possible directions
to analyze algorithms with a non-zero drift term.

\section{Preliminaries}
\label{sec:SR}

Throughout this paper,  $\widehat x = x(1+\delta)$ is the approximation of the real number $x$ under stochastic rounding, with $\abs{\delta} \leq u$ and $u$ is the unit roundoff. 
\mido{We denote $\gamma_n(u) = (1+u)^n - 1$.}
If $x$ is representable, $\widehat x = x$ and $\delta = 0$. 
For a non-representable $x\in \mathbb{R}$, denote $p(x)= \frac{x - \llfloor x \rrfloor}{\llceil x \rrceil - \llfloor x \rrfloor }$, where $\llceil x \rrceil$ is the smallest floating-point number upper than $x$, and $\llfloor x \rrfloor$ is the greatest floating-point number lower than $x$. Note that if $x$ is representable, $x= \llceil x \rrceil = \llfloor x \rrfloor$.
We consider the following stochastic rounding mode, called SR-nearness:

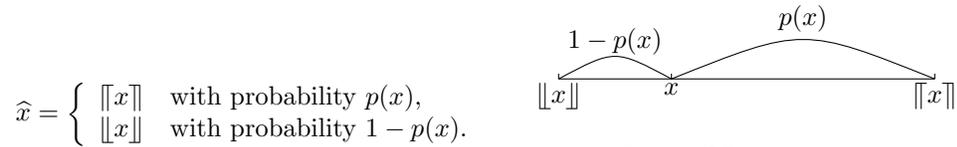
\begin{figure}[h]
    \centering
    \begin{minipage}{0.51\textwidth}
        \begin{align*}
            \widehat x &= \left\{
            \begin{array}{cl}
                \llceil x \rrceil  & \text{with probability $p(x)$,} \\
                \llfloor x \rrfloor & \text{with probability $1-p(x)$.} 
            \end{array} 
            \right. 
        \end{align*}
    \end{minipage}
    \hfill
    \begin{minipage}{0.475\textwidth}
        \centering
        \begin{tikzpicture}[xscale=5]
            \draw (0,0) -- (1,0);
            \draw[shift={(0,0)},color=black] (0pt,0pt) -- (0pt, 2pt) node[below] {$\llfloor x \rrfloor$};
            \draw[shift={(1,0)},color=black] (0pt,0pt) -- (0pt, 2pt) node[below] {$\llceil x \rrceil$};
            \draw[shift={(.3,0)},color=black] (0pt,0pt) -- (0pt, 2pt) node[below] {$x$};
            \draw (0,0) .. controls (.15,.4) .. (.3,0)  (0.15,0.5) node {$1-p(x)$};
            \draw (.3,0) .. controls (.65,.7) .. (1,0) (0.65,0.8) node {$p(x)$} ;
        \end{tikzpicture}
        \caption{\textbf{SR-nearness}.}
    \end{minipage}
\end{figure}
The rounding SR-nearness mode is unbiased (which does not mean that a sequence of operations using SR is necessarily unbiased; for instance, squaring an unbiased error leads to a bias due to the square term that corresponds to a variance):
\begin{align*}
\mathbb{E}(\widehat x) & = p(x)\llceil x \rrceil +(1-p(x))\llfloor x \rrfloor \\
& = p(x)(\llceil x \rrceil - \llfloor x \rrfloor) + \llfloor x \rrfloor =x.
\end{align*}

The following lemma has been proven in~\cite[lem 5.2]{theo21stocha} and shows that rounding errors under SR-nearness are mean independent.

\begin{lemma}\label{lem:meanindp}
    Let $a$ and $b$ be the result of $k-1$ scalar operations and $\delta_1, \ldots, \delta_{k-1}$ be the rounding errors obtained using SR-nearness. Consider $c \leftarrow a \op b$ for $\op \in \{+, -, \times, /\}$, and $\delta_k$ the error of the $k^{\text{th}}$ operation, that is to say, $\hat c = (a \op b) (1+\delta_k)$. The $\delta_k$ are random variables with mean zero and $(\delta_1, \delta_2, \ldots)$ is mean independent, i.e., $ \forall k \geq 2, \mathbb{E}[\delta_k / \delta_1,\ldots , \delta_{k-1}] = \mathbb{E}(\delta_k)$.
\end{lemma}

\mido{When a random variable $X$ is fully determined by a random variable $Y$, the conditional expectation~\cite[sect.~34]{azum} satisfies the following property.}
\begin{proposition}
\mido{Let $X$ and $Y$ be random variables defined on the same probability space. If $X$ is entirely determined by $Y$, meaning there exists a measurable function $g$ such that $X = g(Y)$, then the conditional expectation of $X$ given $Y$ satisfies $\mathbb{E}(X \mid Y) = X$.
}
\end{proposition}

\begin{definition}[\mido{{\cite[p.~295]{azum}}}]
	\label{def:martingale}
	A sequence of random variables $M_1,\cdots, M_n$ is a martingale with respect to the sequence $X_1,\cdots, X_n$ if, for all $k,$
	\begin{itemize}
		\item $M_k$ is a function of $X_1,\cdots, X_k$,
		\item $\mathbb{E}(\lvert M_k \rvert ) < \infty,$ and
		\item $\mathbb{E}[M_k /  X_1,\cdots, X_{k-1}]=M_{k-1}$.
	\end{itemize}

\end{definition}

\begin{lemma}[Azuma--Hoeffding inequality, \mido{{\cite[p.~303]{azum}}}]
\label{lem:azuma}
	Let $M_0,\cdots, M_n$ be a martingale with respect to a sequence $X_1,\cdots, X_n.$ We assume that there exist $a_k<b_k$ such that $a_k \leq M_k - M_{k-1} \leq  b_k$ for $k = 1: n.$ Then, for any $A  > 0$ 
	$$ \mathbb{P}(\lvert M_n - M_0 \rvert \geq A) \leq 2 \exp \left( 
	-\frac{2A^2}{\sum_{k=1}^n(b_k-a_k)^2} 
	\right).
	$$
	In the particular case $a_k=-b_k$ and $\lambda = 2 \exp \left( 
	-\frac{A^2}{2\sum_{k=1}^n b_k^2} \right) $ we have 
	$$ \mathbb{P}\left( \lvert M_n - M_0 \rvert \leq \sqrt{\sum_{k=1}^n b_k^2} \sqrt{2 \ln (2 / \lambda)} \right) \geq 1- \lambda,
	$$
	where $0< \lambda <1$.
\end{lemma}

\begin{remark}
In Lemma~\ref{lem:azuma}, if all $b_k$ are constant with value $b$, 
$$	
\sqrt{\sum_{k=1}^n b_k^2} = \sqrt{\sum_{k=1}^n b^2} = \abs{ b} \sqrt{n}.
$$
\end{remark}

It has been shown~\cite{positive, arar2022stochastic, el2023bounds, ilse, hallman2022precision, theo21stocha} that the mean independence property is sufficient to improve the error analysis of algorithms with SR-nearness. It leads to obtaining a martingale (Definition~\ref{def:martingale}), which is a sequence of random variables such that the expected value of the next value in the sequence, given all the past values, is equal to the current value. Using the Azuma-Hoeffding inequality~\cite[p.~303]{azum}, allows to obtain probabilistic bounds on the error in $\mathcal{O}(\sqrt{n}u)$. For further details, we refer to~\cite[chap 4]{elararphd}.

\section{Errors in sum-product computation graphs} 
\label{sec:sum-product-trees}
In this section, by induction, and for any computation, we build a martingale, the last term of which is the rounding error of the computation. This construction gives the length of this martingale, as well as a condition number based on a deterministic bound on the martingale steps. Together, these quantities allow to apply the Azuma-Hoeffding inequality, or to compute the variance of the error, and thus to probabilistically bound the rounding error of the computation. 

This martingale generalizes the ones found for the recursive summation~\cite{theo21stocha}, the dot product~\cite{ilse}, the pairwise summation~\cite{hallman2022precision, el2023bounds}, and the Horner’s polynomial evaluation~\cite{positive}. This construction applies to any numerical scheme based on additions and products, in which no two variables sharing a rounding error are multiplied. Seeing the computation as a DAG, the parents of a multiplication node cannot share a common ancestor (except for inputs, which are not affected by an error). 
In the case a common ancestor exists for multiplication nodes, then a bias appears (the expectation of the squared error is not zero), which this method cannot account for. We propose in the last section a method to deal with such biases. 


The construction differs according to its last operation. For a sum, the martingale is basically the weighted sum of the two martingales associated to the summands, with one additional term for the last error: the length of the resulting martingale is the length of the longest of the two plus one, and the bound on the step is the weighted mean of the two bounds with the values of the summands as coefficients. 

For a product, the martingale is built by ordering the two martingales of the multiplied terms, and adding one term accounting for the last error. This ordering requires that no individual rounding error is shared in both terms, which forbids that any part of the two terms depend on the same computation, as previously stated.
The length of the resulting martingale is the sum of the lengths of the two plus one, and the bound on the step is the product of the two bounds associated to the operands. 


In Figure~\ref{fig:last_op}, we consider an algorithm in which $\mathtt z$ is the return value, and the last operation is $\mathtt z\leftarrow\mathtt x \op \mathtt y$:
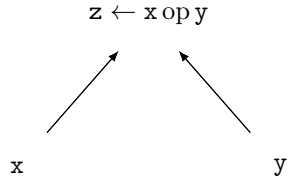
\begin{figure}[h]
\centering
\begin{tikzpicture}[
	level distance=2cm,
	level 1/.style={sibling distance=3.5cm},
	level 2/.style={sibling distance=2cm},
	every node/.style={draw=none, minimum size=9mm},
	edge from parent/.style={<-,draw},
	>=latex
	]
	\node (root) {$\mathtt z\leftarrow\mathtt x \op \mathtt y$}
	child {node (x) {$\mathtt x $}}
	child {node (y) {$\mathtt y $}};
	
	\draw[->] (x) -- (root);
	\draw[->] (y) -- (root);
\end{tikzpicture}
	\caption{Last operation in the computation of a variable $z$, $\op \in \{+, -, \times\}$.}
	\label{fig:last_op}
\end{figure}

\mido{This paper uses the relative error, which is not defined when the result is zero; therefore we assume that  $ x \op \mathtt y \neq 0$.}

\subsection{Base case: \texorpdfstring{$\mathtt{z}$}{} is an input}
\label{sub-sec:base-case}

The base case is straightforward. Since we assume that inputs are exact (void computations), the error is $0$, and it can be seen the last term of the trivial martingale consisting of the empty sequence. The length of this martingale is $0$, and the associated condition number is $1$.

\subsection{Addition}

Suppose that the last operation in the computation of the variable $z$ is an addition, i.e, $\mathtt z\leftarrow\mathtt x + \mathtt y$. Consider the relative errors $\Phi$, $\mathrm{X}$ and $\Psi$ associated respectively to $\mathtt x$, $\mathtt y$, and $\mathtt z$. Note $x$, $y$, and $z$ their respective exact values, and $\hat x$, $\hat y$ and $\hat z$ their computed values. Therefore:  

\begin{equation}
\label{eq=z+}
    \begin{cases}
	\hat{x} = x(1+\Phi), \\
	\hat{y} = y(1+\mathrm{X}), \\
	\hat{z} = z(1+\Psi).
    \end{cases}
\end{equation}
We have $z = x+y$, then, there exists $\delta$ such that $\hat z = (\hat{x} + \hat{y}) (1+\delta)$. Hence,
\begin{align}
	 \Psi &= \frac{\hat z - z}{z} \notag \\
	 &= \frac{ \hat{x} + \hat{y}}{x+y} (1+\delta) -1 \notag \\
	 &= \mido{\frac{ x(1+\Phi) + y(1+\mathrm{X})}{x+y} (1+\delta) -1} \notag \\
	 &= \left(1+ \frac{x}{x+y}\Phi+\frac{y}{x+y} \mathrm{X}\right)(1+\delta) -1. \label{eq:psi}
\end{align}

Suppose by induction that there exist constants $\mathcal{K}_x \geq 1$ (bounding the condition number in the computation of $x$) and $\mathcal{K}_y \geq 1$ (bounding the condition number in the computation of $y$), and martingales $(\Phi_i)_{i=0}^{k-1}$ and $(\mathrm{X}_i)_{i=0}^{l-1}$
\mido{such that:
\begin{equation}
\label{eq:induction-sum}
    \begin{cases}
    \abs{\Phi_{i} - \Phi_{i-1}} &\le \mathcal{K}_xu(1+u)^{i-1},  \ \Phi_0=0, \ \text{and} \ \Phi=\Phi_{k-1},\\
    \abs{\mathrm{X}_{i} - \mathrm{X}_{i-1}} &\le \mathcal{K}_yu(1+u)^{i-1}, \ \mathrm{X}_0=0, \ \text{and} \ \mathrm{X}=\mathrm{X}_{l-1}.
\end{cases}
\end{equation}}
When one of $\hat{x}$ or $\hat{y}$ is exact, as mentioned in~\Cref{sub-sec:base-case}, we assume that the length of the martingale is 0 and the condition number is 1.

\begin{lemma}
	Let $m = \max\{k,l\} + 1$. The stochastic process  $(\Psi_i)_{i=0}^{m-1}$ such that $\Psi_{m-1}=\Psi$, and for all $0\leq i<m-1$,
	$$\Psi_i= \frac{x}{x+y}\Phi_{\min\{i, k\}}+\frac{y}{x+y}\mathrm{X}_{\min\{i, l\}},$$	 
	forms a martingale. 
\end{lemma}

\begin{proof}
	Without loss of generality, let us assume that $k\leq l$. Then, $m=l+1$ and
	   \[
	\begin{cases}
		\Psi_i &= \frac{x}{x+y}\Phi_{i}+\frac{y}{x+y}\mathrm{X}_{i} \ \text{for all} \ 0\leq i\leq k-2 \\		
		\Psi_i &= \frac{x}{x+y}\Phi +\frac{y}{x+y}\mathrm{X}_{i} \ \text{for all} \ k-1 \leq i\leq l-1.\\		
	\end{cases}
	\] 
	Note that $(\Phi_{i})_{i=0}^{m-2}$ with $\Phi_i = \Phi$ for all $k-1 \leq i\leq m-2$ and $(\mathrm{X}_{i})_{i=0}^{m-2}$ are martingales  by induction hypothesis. Since the martingale set is a vector space, as a linear combination of them, $(\Psi_{i})_{i=0}^{m-2}$ is a martingale. \mido{Let us verify for $\Psi_{m-1}$. Using \eqref{eq:psi} and the} mean independence (Lemma~\ref{lem:meanindp}) of $\delta$ from $\Phi$ and $\mathrm{X}$ we have
	\begin{align*}
		 \mathbb{E}[\Psi_{m-1}/\Psi_{m-2}] &= \mathbb{E}\left[\left(1+ \frac{x}{x+y}\Phi+\frac{y}{x+y} \mathrm{X}\right)(1+\delta) -1/\Psi_{m-2}\right]\\
		 &=  \left(1+ \frac{x}{x+y}\Phi+\frac{y}{x+y} \mathrm{X}\right)\mathbb{E}[(1+\delta)/\Psi_{m-2}] -1\\
		 &= \mido{\Psi_{m-2} + \left(1+ \frac{x}{x+y}\Phi+\frac{y}{x+y} \mathrm{X}\right) \mathbb{E}[\delta/\Psi_{m-2}]}\\
		 &=\Psi_{m-2}.
	\end{align*} 
Thus, $(\Psi)_{i=0}^{m-1}$ is a martingale. 

\end{proof}

In this lemma, we have built a martingale by induction when the last operation is an addition. 
In order to use the Azuma-Hoeffding inequality (Lemma~\ref{lem:azuma}), we have to bound martingale increments.

\begin{lemma}
\label{lem:condition-plus}
	Let $\mathcal{K}_z=\frac{\abs{x}}{\abs{x+y}}  \mathcal{K}_x+\frac{\abs{y}}{\abs{x+y}} \mathcal{K}_y$. The martingale $(\Psi_i)_{i=0}^{m-1}$ satisfies
		$$ \abs{\Psi_{i} - \Psi_{i-1}} \leq u C_i,
		$$
		where $C_i=\mathcal{K}_z (1+u)^{i-1}$ for all $1\leq i\leq m-1$.
\end{lemma}
\begin{proof}
  \textcolor{blue}{in appendix~\ref{sec:ap-addition}.}
\end{proof}

\begin{theorem}
	\label{thm:1}
	For all $0 < \lambda <1$, the computed $\hat{z}$ in Equation~\eqref{eq=z+} satisfies under SR-nearness
	\begin{equation}
		\frac{\abs{\hat z -z}}{\abs{z}} \leq \mathcal{K}_z \sqrt{u\gamma_{2(m-1)}(u)} \sqrt{\ln(2/ \lambda)} \mido{= \mathcal{O}(u \sqrt{m})},
	\end{equation}
with probability at least $1-\lambda$, \mido{where $\gamma_m(u) = (1+u)^m -1 = \mathcal{O}(mu)$ }
\end{theorem}
\begin{proof}
  \textcolor{blue}{in appendix~\ref{sec:ap-addition}.}
\end{proof}

\subsection{Multiplication}

Suppose now that the last operation is a multiplication, i.e, $\mathtt z\leftarrow\mathtt x \times \mathtt y$. Consider the relative errors $\Phi$, $\mathrm{X}$, and $\Psi$ associated respectively to $x$, $y$, and $z$. Note $x$, $y$, and $z$ their respective exact values, and $\hat x$, $\hat y$ and $\hat z$ their computed values, with

\begin{equation}
    \label{eq:z*}
    \begin{cases}
	\hat{x} = x(1+\Phi), \\
	\hat{y} = y(1+\mathrm{X}), \\
	\hat{z} = z(1+\Psi).
    \end{cases}
\end{equation}
We have $z = x\times y$, then there exists $\delta$ such that $\hat z = \hat{x} \times \hat{y} (1+\delta)$. Hence,
\begin{align*}
	\Psi &= \frac{\hat z - z}{z} \\
	&= \frac{ \hat{x} \times \hat{y}}{x\times y} (1+\delta) -1\\
	&= (1+\Phi)(1+\mathrm{X}) (1+\delta) -1.
\end{align*} 

We know by induction that $\Phi$ and $\mathrm{X}$ are the last terms of two martingales. However, the multiplication of two martingales is not necessarily a martingale. Consequently, in contrast to the addition case, we have to decide a scheduling of operations in the construction of the martingale $\Psi$. All are equivalent and lead to the same final result.

As presented in~\Cref{lem:multi-sr}, we assume that in figure~\ref{fig:last_op}, the left sub-tree is computed before the right sub-tree, which means that in the computation of $x$, we assume that we don't have any operation on $y$. Consider two martingales $(\Phi_i)_{i=0}^{k-1}$  and $(\mathrm{X}_i)_{i=0}^{l-1}$ such that $\Phi_0=0$, $\Phi=\Phi_{k-1}$, \mido{$\mathrm{X}_0=\ldots=\mathrm{X}_{k-1}=0$,}  $\mathrm{X}=\mathrm{X}_{l-1}$, and random errors in $\Phi$ are different from those of $\mathrm{X}$ (thanks to the multi-linearity of errors in the computation of $z$). The following lemma shows that $\Psi$ is the last term of a martingale built from $(\Phi_i)_{i=0}^{k-1}$  and $(\mathrm{X}_i)_{i=0}^{l-1}$.


\begin{lemma}
\label{lem:multi-sr}
	The stochastic process  $(\Psi_i)_{i=0}^{m-1}$ such that 
		\[\Psi_{i}=
	\begin{cases}	
		\Phi_i =(1+\Phi_i)(1+0)  -1 & \text{for all} \ 0 \leq i \leq k-1 \\
		(1+\Phi)(1+\mathrm{X}_{i-k})  -1 & \text{for all} \ k \leq i \leq m-2 \\
		(1+\Phi)(1+\mathrm{X}) (1+\delta) -1  & \text{for} \ i=m-1, \\
	\end{cases}
	\] 
	forms a martingale.
\end{lemma}

\begin{proof}
	 For all $0\leq i \leq k-1$, by construction of $\Psi$, we have $\Psi_{i}= \Phi_{i}$. Since $(\Phi_i)_{i=0}^{k-1}$ is a martingale, we have
	 $$\mathbb{E}[\Psi_i/\Psi_{i-1}]= \mathbb{E}[\Phi_i/\Psi_{i-1}] = \Phi_{i-1}=\Psi_{i-1}.$$
	   Moreover, for the $k^{\text{th}}$ term we have
	\begin{align*}
		\mathbb{E}[\Psi_k/\Psi_{k-1}] &= \mathbb{E}[(1+\Phi)(1+\mathrm{X}_1)  -1/\Psi_{k-1}] \\
		&= (1+\Phi) \mathbb{E}[(1+\mathrm{X}_1)/\Psi_{k-1}] -1\\
		&= \mido{ (1+\Phi)(1+\mathrm{X}_0) -1}, \quad \mido{(\mathrm{X}_i)_{i=0}^{l-1} \ \text{is a martingale},} \\
		&= (1+\Phi)  -1,  \quad \mido{\mathrm{X}_0=0,} \\
		 &= \Phi.
	\end{align*}
	Since  $(\mathrm{X}_{i-k})_{i=k}^{m-2}$ is a martingale, for all $k < i \leq m-2$, 
	\begin{align*}
		\mathbb{E}[\Psi_i/\Psi_{i-1}] &= \mathbb{E}[(1+\Phi)(1+\mathrm{X}_{i-k})  -1/\Psi_{i-1}]\\
		 &= (1+\Phi)\mathbb{E}[(1+\mathrm{X}_{i-k})/\Psi_{i-1}]  -1 \\
		 &= (1+\Phi)(1+\mathrm{X}_{i-k-1})  -1 = \Psi_{i-1}.
	\end{align*}
	By mean independence of $\delta$ and $\Psi_{m-2}$, we get
	\begin{align*}
		\mathbb{E}[\Psi_{m-1}/\Psi_{m-2}] &= \mathbb{E}[(1+\Phi)(1+\mathrm{X}_{l-1})(1+\delta)  -1/\Psi_{m-2}]\\
		 &= (1+\Phi)(1+\mathrm{X}_{l-1}) \mathbb{E}[(1+\delta)/\Psi_{m-2}]  -1 \\
		 &= \mido{(1+\Phi)(1+\mathrm{X}_{l-1}) -1}\\
		 &= \Psi_{m-2}.
	\end{align*}	
\end{proof}

In order to use Azuma-Hoeffding inequality (Lemma~\ref{lem:azuma}), we need to bound the martingale increments.
We can show by induction that there exist constants $\mathcal{K}_x \geq 1$ (bounding the condition number in the computation of $x$) and $\mathcal{K}_y\geq 1$ (bounding the condition number in the computation of $y$), such that the $i^{\text{th}}$ steps \mido{satisfy 
\begin{equation}
\label{eq:induction-mult}
    \begin{cases}
    \abs{\Phi_{i} - \Phi_{i-1}} &\le \mathcal{K}_xu(1+u)^{i-1} \quad \text{for all} \ 0 \leq i \leq k-1 \\
    \abs{\mathrm{X}_{i} - \mathrm{X}_{i-1}} &\le \mathcal{K}_y u(1+u)^{i-1} \quad \text{for all} \ k \leq i \leq m-1, \\
\end{cases}
\end{equation}
} because $\mathrm{X}_{j}=0$  for all $0\leq j \leq k-1$.

\begin{lemma}
\label{lem:condition-product}
	Let $\mathcal{K}_z= \mathcal{K}_x \mathcal{K}_y$. The martingale $(\Psi_i)_{i=0}^{m-1}$ satisfies
	$$ \abs{\Psi_{i} - \Psi_{i-1}} \leq u C_i,
	$$
	where $C_i=\mathcal{K}_z (1+u)^{i-1}$ for all $1\leq i\leq m-1$.
\end{lemma}
\begin{proof}
  \textcolor{blue}{in appendix~\ref{sec:ap-multiplication}.}
\end{proof}

\begin{theorem}
	\label{thm:2}
	For all $0 < \lambda <1$, the computed $\hat{z}$ in Equation~\eqref{eq:z*} satisfies under SR-nearness
	\begin{equation}
		\frac{\abs{\hat z -z}}{\abs{z}} \leq \mathcal{K}_z \sqrt{u\gamma_{2(m-1)}(u)} \sqrt{\ln(2/ \lambda)} \mido{= \mathcal{O}(u \sqrt{m})},
	\end{equation}
	with probability at least $1-\lambda$.
\end{theorem}
\begin{proof}
  \textcolor{blue}{in appendix~\ref{sec:ap-multiplication}.}
\end{proof}

Theorems~\ref{thm:1} and~\ref{thm:2} show that the error of any algorithm based on elementary operations $\{+, -, \times\}$ and with multi-linear errors has a probabilistic bound in $\mathcal{O}(\sqrt{n}u)$, where $n$ is the number of operations.

\section{Error analysis using the proved generalization}
\label{sec:examples}

In this section, we apply this generalization to algorithms based on elementary operations $\{+, -, \times\}$ with multi-linear errors. First, we consider the pairwise summation algorithm that involves only additions, and we show how this method computes the martingale's length and the condition number. We obtain the same result proved in~\cite{el2023bounds, hallman2022precision} for this algorithm. Next, we analyze the Horner algorithm that combines additions and multiplications, which illustrates the effect of multiplication on the martingale length. We obtain the same result proved in~\cite{positive} for this algorithm. Finally, we examine the Karatsuba algorithm, which demonstrates the flexibility of this method in handling DAGs where, in the case of multiplication, two nodes do not share errors. 

\subsection{Pairwise summation}
\label{sub-sec:pariwise}
We investigate the forward error made by the pairwise summation algorithm under SR. Section~\ref{sec:sum-product-trees} demonstrates that the error generated by this algorithm forms a martingale. In the following, we illustrate how the generalization presented in the previous section can be applied to compute the length of this martingale and bound the condition number. We thus use the Azuma-Hoeffding inequality to compute a probabilistic bound for the error. For illustrative purposes, let’s consider $z= \sum_{i=1}^{n}  x_i$ such that ($\lceil n/2 \rceil$ is the smallest integer more than or equal to $n/2$):

\begin{center}
\scalebox{0.9}{
\begin{tikzpicture}[
	level 1/.style={sibling distance=60mm},
	level 2/.style={sibling distance=30mm},
	level 3/.style={sibling distance=20mm},
	level 4/.style={sibling distance=10mm},
	every node/.style={draw=none, minimum size=7mm},
	edge from parent/.style={<-,draw},
	>=latex
	]
	\node (root) {$m=h; +; \mathcal{K} = \frac{\sum_{i=1}^{n} \abs{ x_i}}{\abs{ \sum_{i=1}^{n} x_i }}$}
	child {node {$m=h-1; +; \frac{\sum_{i=1}^{\lceil n/2 \rceil } \abs{ x_i}}{\abs{ \sum_{i=1}^{\lceil n/2 \rceil} x_i }}$}
		child {node {\textcolor{blue}{$\ldots$}}
			child {node {$m=1; +; \frac{\abs{x_1}+\abs{x_{2}}}{\abs{x_1 + x_{2}}}$}
				child {node {$x_1$}}
				child {node {$x_2$}}
			}
			child {node {}
				child {node {}}
				child {node {}}
			}
		}
		child {node {\textcolor{blue}{$\ldots$}}
			child {node {}
				child {node {}}
				child {node {}}
			}
			child {node {}
				child {node {}}
				child {node {}}
			}
		}
	}
	child {node {$m=h-1; +;  \frac{\sum_{i=\lceil n/2 \rceil +1}^{n} \abs{ x_i}}{\abs{ \sum_{i=\lceil n/2 \rceil +1}^{n} x_i }}$}
		child {node {\textcolor{blue}{$\ldots$}}
			child {node {}
				child {node {}}
				child {node {}}
			}
			child {node {}
				child {node {}}
				child {node {}}
			}
		}
		child {node {\textcolor{blue}{$\ldots$}}
			child {node {}
				child {node {}}
				child {node {}}
			}
			child {node {$m=1; +; \frac{\abs{x_{n-1}}+\abs{x_{n}}}{\abs{x_{n-1} + x_{n}}}$}
				child {node {$x_{n-1}$}}
				child {node {$x_{n}$}}
			}
		}
	};
	
\end{tikzpicture}
	}
\end{center}
At each internal node, we have:
\begin{itemize}
    \item On the left, $m$ represents the martingale length. In this case, $m = 1, 2, \ldots,h$, where $h$ is the height of \mido{the} tree.
    \item In the middle, we have the elementary operation between the two children. In this case, only additions are considered.
    \item On the right, we have the current condition number from the leaves up to this node.
\end{itemize}

\mido{
\begin{theorem}
    For all $0 < \lambda <1$ and $z= \sum_{i=1}^{n}  x_i$, the computed $\hat{z}$ satisfies under SR
\begin{equation}
    \label{pairwise:martingale-bound}
    \frac{\abs{\hat{z} - z}}{\abs{z}} \leq \mathcal{K} \sqrt{u\gamma_{2h}(u)} \sqrt{\ln(2/ \lambda)} =\mathcal{O}(\sqrt{h}u),
\end{equation}
with probability at least $1-\lambda$, where $h= \lfloor \log_2(n) \rfloor$.
\end{theorem}}

\begin{proof}
Since there are only additions, $m$ is $\max\{m_l, m_r\} + 1$, where $m_l$ and $m_r$ are the martingale lengths at the left and right sub-trees, respectively. We assume that the inputs are exact, so $m=0$ at each leaf. Consequently, since we add one at each step, at the root, $m=h$, the height of the tree.

Let us compute the condition number of the root. From Lemma~\ref{lem:condition-plus}, we have $\mathcal{K}=\frac{\abs{x}}{\abs{x+y}}  \mathcal{K}_x+\frac{\abs{y}}{\abs{x+y}} \mathcal{K}_y$ where $x = \sum_{i=1}^{\lceil n/2 \rceil}  x_i$, $y = \sum_{i=\lceil n/2 \rceil +1}^{n}  x_i$, and 
\begin{align*}
      \mathcal{K}_x = \frac{\sum_{i=1}^{\lceil n/2 \rceil } \abs{ x_i}}{\abs{ \sum_{i=1}^{\lceil n/2 \rceil} x_i }} = \frac{\sum_{i=1}^{\lceil n/2 \rceil } \abs{ x_i}}{\abs{x}}, \qquad
        \mathcal{K}_y = \frac{\sum_{i=\lceil n/2 \rceil +1}^{n} \abs{ x_i}}{\abs{ \sum_{i=\lceil n/2 \rceil +1}^{n} x_i }} = \frac{\sum_{i=\lceil n/2 \rceil +1}^{n} \abs{ x_i}}{\abs{y}}.
\end{align*}
We thus have $x+y  = \sum_{i=1}^{n}  x_i$ and
\begin{align*}
    \mathcal{K} &=\frac{\abs{x}}{\abs{x+y}}  \mathcal{K}_x+\frac{\abs{y}}{\abs{x+y}} \mathcal{K}_y
    = \frac{1}{\abs{\sum_{i=1}^{n}  x_i}} \left( \abs{x}\mathcal{K}_x + \abs{y} \mathcal{K}_y \right)\\
    &= \frac{1}{\abs{\sum_{i=1}^{n}  x_i}} \left( \sum_{i=1}^{\lceil n/2 \rceil } \abs{ x_i} + \sum_{i=\lceil n/2 \rceil +1}^{n} \abs{ x_i} \right)
    = \frac{\sum_{i=1}^{n} \abs{ x_i}}{\abs{ \sum_{i=1}^{n} x_i }}.
\end{align*}
Finally, Theorem~\ref{thm:1} shows that 
\begin{equation*}
    \frac{\abs{\hat{z} - z}}{\abs{z}} \leq \mathcal{K} \sqrt{u\gamma_{2h}(u)} \sqrt{\ln(2/ \lambda)} =\mathcal{O}(\sqrt{h}u),
\end{equation*}
with probability at least $1-\lambda$.
\end{proof}
Interestingly, the bound in~(\ref{pairwise:martingale-bound}) is identical to the bound proved in~\cite{el2023bounds} for the pairwise summation using the AH method.
This proof can easily be adapted to any summation tree, leading to a bound of $\sqrt{h}u$ with high probability, with $h$ the height of the tree. 
    
\subsection{Horner algorithm}
\label{sub-sec:horner}
The previous example only had additions. \mido{Let us now apply the results of} Section~\ref{sec:sum-product-trees} to Horner's polynomial evaluation, with both additions and multiplication.  
Let $P(x) = \sum_{i=0}^n a_i x^i$, Horner's algorithm consists in writing this polynomial as 
$P(x)= (((a_nx +a_{n-1})x +a_{n-2})x \ldots +a_1)x +a_0.
$
\begin{center}
\begin{tikzpicture}[
    level 1/.style={sibling distance=24mm},
    level 2/.style={sibling distance=17mm},
    level 3/.style={sibling distance=20mm},
    level 4/.style={sibling distance=17mm},
    level 5/.style={sibling distance=17mm},
    every node/.style={draw=none, minimum size=7mm},
    edge from parent/.style={<-,draw},
    >=latex
    ]
    \node (root) {$m=2n$; $+$; $\frac{\sum_{i=0}^{n} \abs{ a_i x^i }}{\abs{ \sum_{i=0}^{n}  a_i x^i }}$}
    child {node {$2n-1$; $\times$; $\frac{\sum_{i=1}^{n} \abs{ a_i x^i }}{\abs{ \sum_{i=1}^{n}  a_i x^i }}$} 
        child {node {$\ldots$}
            child {node {$2$; $+$; $\frac{\abs{a_n x}+\abs{a_{n-1}}}{\abs{a_n x+ a_{n-1}}}$}
                child {node {$1$; $\times$; $1$}
                    child {node {$a_n$}}
                    child {node {$x$}}
                }
                child {node {$a_{n-1}$}}
            }
            child {node {$x$}}
        }
        child {node {$\ldots$}}
    }
    child {node {$a_0$}};
\end{tikzpicture}
\end{center}
As in the previous example, internal nodes represent three elements. They show the martingale length on the left, the operation between child nodes in the middle, and the condition number from the leaves to the node on the right.
\mido{
\begin{theorem}
    For all $0 < \lambda <1$, the computed $\hat{P}(x)$ satisfies under SR
\begin{equation}
\label{horner}
    \frac{\abs{\hat{P}(x) - P(x)}}{\abs{P(x)}} \leq \mathcal{K} \sqrt{u \gamma_{4n}(u)} \sqrt{\ln(2/ \lambda)} = \mathcal{O}(\sqrt{n}u),
\end{equation}
with probability at least $1-\lambda$.
\end{theorem}
}

\begin{proof}
Let recall that $m$ is $\max\{m_l, m_r\} + 1$ for additions and $m_l + m_r + 1$ for multiplications, where $m_l$ and $m_r$ are the martingale lengths at the left and right sub-trees, respectively. We suppose that $m=0$ for leaves. In Horner's algorithm, $m_r=0$, so that both in additions and multiplications, $m=m_l+1$. Since there are $n$ additions and $n$ multiplications, we have $m=2n$. 

Let us compute the condition number bound. The first operation is a multiplication between $a_n$ and $x$. According to Lemma~\ref{lem:condition-product}, the condition number is $1$. The second operation is an addition between $a_n x$ and $a_{n-1}$. According to Lemma~\ref{lem:condition-plus}, the condition number is $\frac{\abs{a_n x} + \abs{a_{n-1}}}{\abs{a_n x + a_{n-1}}}$. For the root, we \mido{sum $a_0$ and $\sum_{i=1}^{n}  a_i x^i$, then} 
\begin{align*}
    \mathcal{K} &=\frac{\abs{\sum_{i=1}^{n}  a_i x^i}}{\abs{\sum_{i=1}^{n}  a_i x^i + a_0}}  \frac{\sum_{i=1}^{n} \abs{ a_i x^i }}{\abs{ \sum_{i=1}^{n}  a_i x^i }} +\frac{\abs{a_0}}{\abs{\sum_{i=1}^{n}  a_i x^i + a_0}} \\
    &= \frac{\sum_{i=0}^{n} \abs{ a_i x^i }}{\abs{ \sum_{i=0}^{n}  a_i x^i }}.
\end{align*}
Note that the condition number remains the same in the case of multiplication by an input. Finally, Theorem~\ref{thm:1} and Theorem~\ref{thm:2} show
\begin{equation*}
    \frac{\abs{\hat{P}(x) - P(x)}}{\abs{P(x)}} \leq \mathcal{K} \sqrt{u \gamma_{4n}(u)} \sqrt{\ln(2/ \lambda)} = \mathcal{O}(\sqrt{n}u),
\end{equation*}
with probability at least $1-\lambda$.
\end{proof}
Interestingly, the bound in~(\ref{horner}) is identical to the bound proved in~\cite[thm IV.2]{positive} for the Horner algorithm.

\subsection{Karatsuba polynomial multiplication}
\label{sub-sec:karatsuba}

Karatsuba multiplication \cite{karatsuba1995complexity} is a divide-and-conquer algorithm that reduces the
number of multiplications in the product of two polynomials\footnote{We target polynomials with a number of coefficients that is a power-of-2.}. There are different variants; here, we consider the substractive variant. 

Let us consider two polynomials $A$ and $B$ of degree $2^n-1$.

\begin{itemize}

  \item If $n = 0$, the Karatsuba product of $A$ and $B$ reduces to a scalar multiplication, $K(A, B) = a_0.b_0$.

  \item If $n \ge 1$, we write $A = A_h.X^{2^{n-1}} + A_l$ and $B = B_h.X^{2^{n-1}} + B_l$
        where $A_h$ and $B_h$ capture the high order coefficients and $A_l$ and $B_l$
        capture the low order coefficients of $A$ and $B$ respectively. Then the product of $A$ and $B$ is
        \begin{equation*}
          K(A, B) =  P_2.X^{2^n} + (P_0 + P_1 + P_2).X^{2^{n-1}} + P_0
        \end{equation*}
        where $P_0 = K(A_l,B_l)$, $P_2 = K(A_h, B_h)$ and $P_1 = K(A_h - A_l, B_l - B_h)$.
\end{itemize}

We can note that this recursive step uses only three polynomial multiplications
instead of four in a recursive formulation of the classical multiplication algorithm, leading to a complexity of $O(n^{\log_2(3)})$ instead of $O(n^2)$. 

The following result allows to apply the results proven in the previous section. Figure~\ref{KaraDAGdeg3} illustrates on a product of polynomials of degree 3 how one term of all the multiplications in the algorithm results from computations on $A$, and the other from computations on $B$, which is key to the proof of the theorem. 

\begin{theorem} If $A$ and $B$ result of independent computations, $K(A, B)$ has a martingale-inducing computation DAG.
\end{theorem}
\begin{proof}
  By induction on $n$.

  For $n=0$ then $K(A,B) = a_0.b_0$. The computation DAG is a single multiplication node and $a_0$ and $b_0$ are independent.

  For $n\ge1$, we consider $K(A,B)$ where $A$ and $B$ of degree $2^n-1$.
  First, we will show that the partial products $P_0$, $P_1$ and $P_2$ are computed with martingale-inducing DAGs.

  Because $A$ and $B$ are independent, so are $A_l$ and $B_l$.
  Therefore, $P_0 = K(A_l, B_l)$, where $A_l$ and $B_l$ are of degree $2^{n-1}-1$ has by induction a martingale-inducing DAG.
  By the same reasoning, we show that $P_2 = K(A_h, B_h)$ has a martingale-inducing DAG.

  For $P_1 = K(A_h - A_l, B_l - B_h)$, $A_h$ and $A_l$ are not necessarily independent, but they are combined using a
  subtraction operation. Same for $B_h$ and $B_l$. Moreover, the resulting polynomials $A_h - A_l$ and $B_l - B_h$
  are independent of degree $2^{n-1}-1$. Therefore, by the induction hypothesis, $P_1$ is computed with a martingale-inducing DAG.

  Finally, $K(A,B) = P_2.X^{2^n} + (P_0 + P_1 + P_2).X^{2^{n-1}} + P_0$.
  We can ignore the multiplications by $X^{2^{n}}$ and $X^{2^{n-1}}$, which only shift the position of the coefficients and
  do not introduce numerical errors.

  The coefficients of $K(A,B)$ result of sums of coefficients in $P_0$, $P_1$, and $P_2$. The operands are not always independent
  because some coefficients are shared, for instance, between $P_0$ and $P_0.X^{2^{n-1}}$. Nevertheless, all the operations in the resulting DAG are sums. Therefore, we conclude that the computation of $K(A,B)$ is martingale-inducing.
\end{proof}

  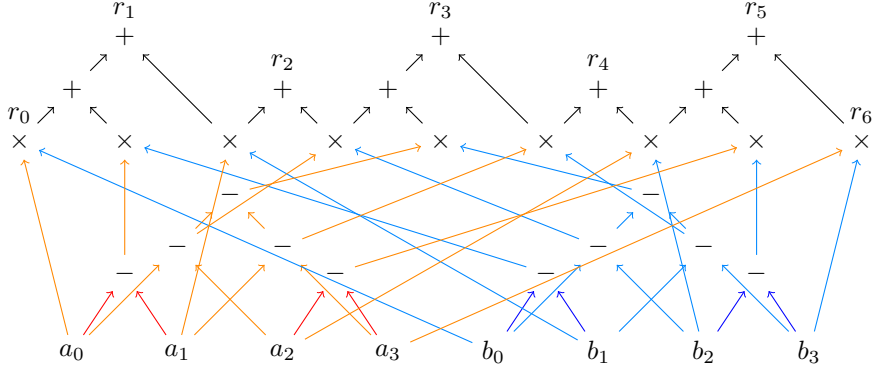
\begin{figure}
    \centering
    \begin{tikzpicture}[scale=0.7]
    \definecolor{red}{HTML}{FF0000}
    \definecolor{lightred}{HTML}{FF8800}
    \definecolor{green}{HTML}{00FF00}
    \definecolor{blue}{HTML}{0000FF}
    \definecolor{lightblue}{HTML}{0088FF}
    \node (a0) at (0,0){$a_0$};
    \node (a1) at (2,0){$a_1$};
    \node (a2) at (4,0){$a_2$};
    \node (a3) at (6,0){$a_3$};
    \node (b0) at (8,0){$b_0$};
    \node (b1) at (10,0){$b_1$};
    \node (b2) at (12,0){$b_2$};
    \node (b3) at (14,0){$b_3$};
    
    \node (m11) at (1,1.5){$-$};
    \node (m12) at (5,1.5){$-$};
    \node (m13) at (9,1.5){$-$};
    \node (m14) at (13,1.5){$-$};
    
    \path[->, red](a0) edge (m11);
    \path[->, red](a1) edge (m11);
    \path[->, red](a2) edge (m12);
    \path[->, red](a3) edge (m12);
    \path[->, blue](b0) edge (m13);
    \path[->, blue](b1) edge (m13);
    \path[->, blue](b2) edge (m14);
    \path[->, blue](b3) edge (m14);

    \node (m21) at (2,2){$-$};
    \node (m22) at (4,2){$-$};
    \node (m23) at (10,2){$-$};
    \node (m24) at (12,2){$-$};

    \path[->, lightred](a0) edge (m21);
    \path[->, lightred](a2) edge (m21);
    \path[->, lightred](a1) edge (m22);
    \path[->, lightred](a3) edge (m22);
    \path[->, lightblue](b0) edge (m23);
    \path[->, lightblue](b2) edge (m23);
    \path[->, lightblue](b1) edge (m24);
    \path[->, lightblue](b3) edge (m24);

    \node (m31) at (3,3){$-$};
    \node (m32) at (11,3){$-$};

    \path[->, lightred](m21) edge (m31);
    \path[->, lightred](m22) edge (m31);
    \path[->, lightblue](m23) edge (m32);
    \path[->, lightblue](m24) edge (m32);
    
    \node (f0) at (-1,4){$\times$};\node (r0) at (-1,4.5){$r_0$};
    \node (f1) at (1,4){$\times$};
    \node (f2) at (3,4){$\times$};
    \node (f3) at (5,4){$\times$};
    \node (f4) at (7,4){$\times$};
    \node (f5) at (9,4){$\times$};
    \node (f6) at (11,4){$\times$};
    \node (f7) at (13,4){$\times$};
    \node (f8) at (15,4){$\times$};\node (r6) at (15,4.5){$r_6$};
    
    \path[->, lightred](a0) edge (f0);
    \path[->, lightblue](b0) edge (f0);
    \path[->, lightred](m11) edge (f1);
    \path[->, lightblue](m13) edge (f1);
    \path[->, lightred](a1) edge (f2);
    \path[->, lightblue](b1) edge (f2);
    \path[->, lightred](m21) edge (f3);
    \path[->, lightblue](m23) edge (f3);
    \path[->, lightred](m31) edge (f4);
    \path[->, lightblue](m32) edge (f4);
    \path[->, lightred](m22) edge (f5);
    \path[->, lightblue](m24) edge (f5);
    \path[->, lightred](a2) edge (f6);
    \path[->, lightblue](b2) edge (f6);
    \path[->, lightred](m12) edge (f7);
    \path[->, lightblue](m14) edge (f7);
    \path[->, lightred](a3) edge (f8);
    \path[->, lightblue](b3) edge (f8);
    
    \node (p1) at (0, 5){$+$};
    \path[->](f0) edge (p1);
    \path[->](f1) edge (p1);
    \node (p2) at (1, 6){$+$};\node (r1) at (1,6.5){$r_1$};
    \path[->](p1) edge (p2);
    \path[->](f2) edge (p2);
    \node (p3) at (4, 5){$+$};\node (r2) at (4, 5.5){$r_2$};
    \path[->](f3) edge (p3);
    \path[->](f2) edge (p3);

    \node (p8) at (12, 5){$+$};
    \path[->](f6) edge (p8);
    \path[->](f7) edge (p8);
    \node (p7) at (13, 6){$+$};\node (r5) at (13,6.5){$r_5$};
    \path[->](p8) edge (p7);
    \path[->](f8) edge (p7);
    \node (p6) at (10, 5){$+$};\node (r4) at (10,5.5){$r_4$};
    \path[->](f6) edge (p6);
    \path[->](f5) edge (p6);
    
    \node (p4) at (6, 5){$+$};
    \path[->](f4) edge (p4);
    \path[->](f3) edge (p4);

    \node (p5) at (7, 6){$+$};\node (r3) at (7,6.5){$r_3$};
    \path[->](p4) edge (p5);
    \path[->](f5) edge (p5);
    \end{tikzpicture}
    \caption{Computation DAG for the Karatsuba multiplication $R=A\times B$ of two polynomials of degree three --- first three levels are the subtractions in the recursive calls, each concerning only one of $A$ or $B$; then all products are performed, with one operand coming from $A$ (blue) and the other from $B$ (red); finally, zero to two levels of additions yield the result --- different hues of blue and red for legibility purpose only.}\label{KaraDAGdeg3}
    \end{figure}

\paragraph{Length of the error martingale}

Let us now compute the length of the martingale for each coefficient of the Karatsuba product.
For $A,B$ of degree $2^n-1$, let $R = K(A,B)$ with degree $d = 2^{n+1}-2$.

\begin{equation*}
  R = r_{2^{n+1}-2}.X^{2^{n+1}-2} + \ldots + r_1.X + r_0
\end{equation*}

The coefficients of $A$ and $B$ can either be constant inputs or result of previous martingale-inducing computations.
We note $m_A$ (respectively $m_B$) the maximum martingale length of the coefficients of $A$ (respectively $B$).
When coefficients are constant, $m_A = m_B = 0$.

\begin{theorem}
  For $i \in \llbracket 0, d \rrbracket$, the length of the error martingale in the computation of coefficient $r_i$ is
  \begin{equation*}
    m(i, d) = 1 + 3 \lfloor\log_2 \min\{i+1, d-i+1\}\rfloor + m_A + m_B.
  \end{equation*}
\end{theorem}
\begin{proof}
  \textcolor{blue}{in appendix~\ref{sec:ap-karatsuba}.}
\end{proof}

\begin{table}
  \centering
  \begin{tabular}{lll}
    \toprule
    $n$ & $d$  & $m(d,d) \ldots m(1,d)\; m(0,d)$      \\
    \midrule
    0 & 0  & 1                                 \\
    1 & 2  & 1 4 1                             \\
    2 & 6  & 1 4 4 7 4 4 1                     \\
    3 & 14 & 1 4 4 7 7 7 7 10 7 7 7 7 4 4 1    \\
    \bottomrule
  \end{tabular}
  \caption{Values of $m(i, d)$ for $n \le 3$ and $m_A = m_B = 0$. \label{tab:mid}}
\end{table}

Let us consider some properties of function $m(i, d)$. Table~\ref{tab:mid} shows the values of $m(i, d)$ for $n \le 3$.
We note that:
\begin{itemize}
  \item \emph{(Property 1)} $m(i,d)$ is symmetric with respect to $d/2$, $m(i,d) = m(d-i,d)$.
  \item \emph{(Property 2)} $m(i,d)$ reaches its maximum for $i = d/2$, $m(d/2, d) = 1 + 3n + m_A + m_B$.
  \item \emph{(Property 3)} The $2^{n-1}$ coefficients to the left and to the right of the central coefficient $d/2$ have the second largest martingale length:\\
        $\forall i \in \rrbracket 2^{n-1} -2, d/2 \llbracket\;\cup\;\rrbracket d/2, d - (2^{n-1}-2) \llbracket$,
        $m(i, d) = 1 + 3(n-1) + m_A + m_B$.
\end{itemize}

It is natural that the error martingale length is smallest for the extreme degrees $r_0$ and $r_d$ since they result from a single product.
On the contrary, the coefficients around $d/2$ are the most sensitive to errors because they result from the sum of many different partial products.

\mido{
\begin{theorem}
  For all $0 < \lambda <1$ and $i \in \llbracket 0, d/2 \rrbracket$, the computed $\hat{r_i}$ satisfies under SR
  \begin{equation}
    \label{eq:r_i}
    \frac{\abs{\hat{r_i} - r_i}}{\abs{r_i}} \leq \mathcal{K}\mido{_i} \sqrt{u \gamma_{2(3\lfloor \log_2(i+1)\rfloor)}(u)} \sqrt{\ln(2/ \lambda)},
    \end{equation}
    with probability at least $1-\lambda$.
\end{theorem}}

\begin{proof}
For all $i \in \llbracket 0, d/2 \rrbracket$, \textcolor{blue}{combining the previous result on the martingale length} with Theorems~\ref{thm:1} and~\ref{thm:2} shows that
\begin{equation*}
\frac{\abs{\hat{r_i} - r_i}}{\abs{r_i}} \leq \mathcal{K}\mido{_i} \sqrt{u \gamma_{2(3\lfloor \log_2(i+1)\rfloor)}(u)} \sqrt{\ln(2/ \lambda)}
\end{equation*}
with probability at least $1-\lambda$.
The bound is maximal for the central coefficient $i=d/2$ (due to \emph{Property 1})  with $d=2^{n+1}-2$, 
\begin{equation}
    \label{eq:r_d/2}
\frac{\abs{\hat{r}_{d/2} - r_{d/2}}}{\abs{r_{d/2}}} \leq \mathcal{K}_{\mido{d/2}} \sqrt{u \gamma_{6n}(u)} \sqrt{\ln(2/ \lambda)}.
\end{equation}
\end{proof}

We perform numerical experiments for $d$ varying from 3 to $2^{16}-1$. The
computation is performed in IEEE-754 RN-binary32 and SR-nearness-binary32.
Errors are computed for the central coefficient $d/2$ against a IEEE-754 binary64 reference. For each degree,
three SR samples are computed with Verificarlo~\cite{verificarlo}. The condition number bound, $\mathcal{K}$, is computed following the lemmas \ref{lem:condition-plus} and \ref{lem:condition-product}.

First, we consider polynomials with positive coefficients uniformly sampled in $[0,1]$ in Figure~\ref{fig:karatsuba-experiment}. The bound growth is dominated by $\mathcal{K}$, which grows with $d$. Despite this, the actual error grows slowly for these inputs and stays under $2^{-20}$.

Then, we consider polynomials with coefficients uniformly sampled in $[-0.5, 0.5]$. The condition number still dominates the bound. When we have both positive and negative coefficients, catastrophic cancellations between terms trigger often, accounting for the faster growth of $\mathcal{K}$.  We observe that the SR and RN samples also show this effect, with a higher error than before: for $d=2^{16}-1$, the relative error is around $2^{-11}$.

\textcolor{blue}{In our analysis, we define $\mathcal{K}$ as an upper bound on the condition number derived from the deterministic bounds on the martingale increments associated with the algorithm’s error process. 
Unlike traditional condition numbers, which are typically defined for mathematical problems and are implementation-agnostic, our $\mathcal{K}$ is algorithm-dependent.}

\textcolor{blue}{For the previous examples, pairwise summation and Horner’s polynomial evaluation, $\mathcal{K}$ coincides with the conventional condition number, and is equal to 1 for positive inputs. However, for Karatsuba polynomial multiplication, $\mathcal{K}$ differs from that of naive multiplication due to the algorithm's structure. Specifically, the cross-product term $P_1$ introduces additional error sources that are later canceled in the final sum. This cancellation does not eliminate their contribution to $\mathcal{K}$, resulting in a $\mathcal{K}$ that grows with the input size. This growth explains the loose bound observed in our experiments, where $\mathcal{K}$ significantly influences both the bound and the error (when coefficients are sampled in $[-0.5, 0.5]$). This highlights how the algorithm's design influences the error analysis and underscores the importance of considering both the mathematical problem and the specific implementation.}

\begin{figure}
    \centering
    \includegraphics[width=0.5\linewidth]{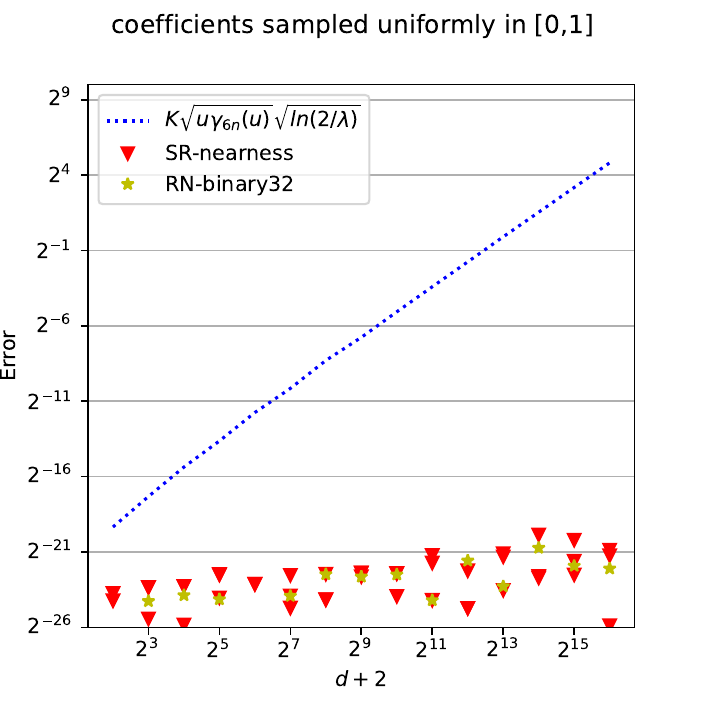}%
    \includegraphics[width=0.5\linewidth]{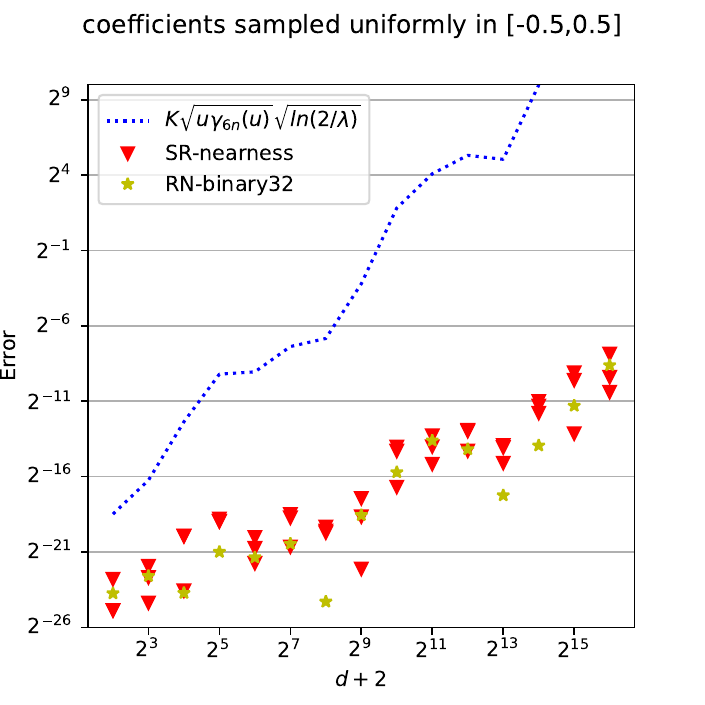}
    \caption{Relative error for the subtractive Karatsuba algorithm. In the left plot, coefficients were uniformly sampled in $[0,1]$. In the right plot, coefficients were uniformly sampled in $[-0.5,0.5]$. ($1-\lambda$ = 0.9 and $d = 2^{n+1}-2$). }
    \label{fig:karatsuba-experiment}
\end{figure}

\section{Doob-Meyer decomposition and non-linear errors}
\label{sec:DM}

To establish a comprehensive framework, we propose using the Doob-Meyer decomposition, a central result in the study of stochastic processes~\cite[p 296]{doob1953stochastic}. This decomposition separates a stochastic process into two distinct components: a martingale part and a predictable process. Let us first recall the definition of a predictable stochastic process~\cite[p 65]{dacunha2012probability}.
\begin{definition}
    Given a filtration $(\mathbb{F}_n)_{n\geq 0}$, a stochastic process $X_n$ is predictable if $X_0$ is $\mathbb{F}_{0}$-measurable, and $X_n$ is $\mathbb{F}_{n-1}$-measurable for all $n\geq 1$.
\end{definition}
This means that the value of $X_n$ is known at the previous time step. Now, let us state the Doob–Meyer decomposition. 

\begin{theorem}[Doob–Meyer decomposition]
	\label{thm:doob}
	Let $(\mathbb{F}_{k})_{0\leq k \leq n}$ and $X_0,\ldots,X_n$ an adapted stochastic process locally integrable, meaning that $\mathbb{E}(\abs{X_k}) < \infty$ for all $0\leq k \leq n$. There exists a martingale $M_0,\ldots,M_n$ and a predictable integrable sequence $A_0,\ldots,A_n$  starting with $A_0 =0$ for which we have:
	\[
	\begin{cases}
		X_n &= M_n + A_n, \\
		
		A_n &= \sum_{k=1}^{n} \mathbb{E}[X_k - X_{k-1}/ \mathbb{F}_{k-1}],\\
		\mathbb{E}(M_n) &=0.
	\end{cases}
	\] 
	This decomposition is almost surely unique.
\end{theorem}

The martingale $M_n$ reflects the information available up to time $n$. It does not exhibit any drift and captures the unbiased random component of the stochastic process $X_n$. While the sequence $A_n$ represents the cumulative effect of the predictable part of the stochastic process $X_n$. It can be interpreted as the drift of $X_n$. Its predictability means that, at each step, one knows the value of the drift at the next step. For instance, at a step when the algorithm squares a value with error, a term square of the current error will be added to the drift, while the martingale remains centered on 0.

We propose to use Doob-Meyer decomposition to analyze the error under SR-nearness.
We consider an algorithm executed under SR-nearness. Its error is a stochastic process $X = \hat{y} - y$.
Because each random error $\delta_i$ is bounded $\abs{\delta_i}\leq u$, the resulting stochastic process $X$  must also be bounded and is locally integrable. Therefore we can apply Doob-Meyer decomposition and write the error as the sum of a martingale and a drift:
\begin{equation}
	\label{eq:general}
	\hat{y} - y = M + A,
\end{equation}

The martingale component in Equation~\eqref{eq:general} captures the unbiased stochastic behavior; in other words, the errors that can be compensated with SR, while the bias is the expected last term of the drift. 

\paragraph{\textbf{Multi-linear error, $A = 0$}}
In this paper, we study algorithms whose computation graphs are martingale-inducing DAGs. For these algorithms, the error terms are always of degree one, which is why we described them as having multi-linear errors. In fact, in a martingale-inducing DAG, the product of two nodes is allowed only if they have different errors, preventing any increase in the degree of the errors. Furthermore, the addition operation does not increase the degree of errors, even if the two operands share some errors.

In the case of multi-linear errors algorithms, as shown in Section~\ref{sec:sum-product-trees}, the forward error is always captured by a martingale, therefore for multi-linear error the drift component $A$ is zero.

\paragraph{\textbf{Non-linear error}}

For non-linear error algorithms, we cannot apply the method from Section~\ref{sec:sum-product-trees}.  Nevertheless, Doob-Meyer decomposition still applies and provides a simplifying framework for analyzing the error.
Indeed, the martingale term can be studied with the Azuma-Hoeffding and has a probabilistic bound in $\mathcal{O}(\sqrt{n}u)$. The problem is, therefore, reduced to the study of the drift term $A$.

In~\cite{el2023bounds}, we have studied variance computation algorithms. By deterministically bounding the drift term $A$, we showed that it was negligible at the first order over $u$ and proved an error bound in $\mathcal{O}(\sqrt{n}u)$. El Arar et al.~\cite[thm 3]{arar2024probabilistic} have implicitly used this decomposition to study the effect of the number of random bits required to implement SR effectively. We conjecture that the drift is negligible when $nu^2=o(1)$. However, for low precision computations, the drift may have a dominant effect on the precision of the result. 

Doob-Meyer provides an interesting decomposition for analyzing non-linear algorithms. Nevertheless, in general, it is not easy to build the decomposition, and bounding the error of general non-linear algorithms under SR-nearness remains an open problem.

\section{Conclusion}
\label{sec:conclusion}

The worst-case error bound for a computation involving $n$ elementary operations is $\mathcal{O}(nu)$.  This bound, while useful, can be overly pessimistic as it assumes a deterministic accumulation of errors and does not account for error compensation phenomena. With SR, the use of probabilistic tools, including martingales, variance analysis, and concentration inequalities, allows us to better investigate rounding errors behavior, establish probabilistic error bounds in $\mathcal{O}(\sqrt{n}u)$.

In this paper, we propose a general methodology to build a martingale for any computation DAG with multi-linear errors arising from addition, subtraction, and multiplication operations. We applied this methodology to pairwise summation and Horner algorithms, confirming results consistent with earlier works on these algorithms under SR. Moreover, to the best of our knowledge, we are the first to analyze Karatsuba polynomial multiplication under SR. Using our approach, we established a probabilistic error bound in $\mathcal{O}(\sqrt{n}u)$ for Karatsuba’s algorithm as well. We have also discussed how to analyze the error of a general algorithm using the Doob-Meyer decomposition that separates the martingale term and the drift part. We believe that this probabilistic framework can serve as an effective tool to improve the rounding error analysis under SR in numerical algorithms.

The scripts to reproduce the numerical experiments of Section~\ref{sub-sec:karatsuba} are made available at \url{https://github.com/verificarlo/sr-karatsuba}.


\bibliographystyle{siamplain}
\bibliography{references}

\begin{thebibliography}{10}

\bibitem{theo21stocha}
{\sc M.~P. Connolly, N.~J. Higham, and T.~Mary}, {\em Stochastic rounding and
  its probabilistic backward error analysis}, SIAM Journal on Scientific
  Computing,  (2021).

\bibitem{survey}
{\sc M.~Croci, M.~Fasi, N.~J. Higham, T.~Mary, and M.~Mikaitis}, {\em
  Stochastic rounding: implementation, error analysis and applications}, Royal
  Society Open Science, 9 (2022), p.~211631.

\bibitem{dacunha2012probability}
{\sc D.~Dacunha-Castelle, D.~McHale, and M.~Duflo}, {\em Probability and
  Statistics: Volume II}, no.~v. 2, Springer New York, 2012.

\bibitem{verificarlo}
{\sc C.~Denis, P.~de~Oliveira~Castro, and E.~Petit}, {\em Verificarlo: Checking
  floating point accuracy through \text{M}onte \text{C}arlo arithmetic}, in
  23nd {IEEE} Symposium on Computer Arithmetic, {ARITH} 2016, Silicon Valley,
  CA, USA, July 10-13, 2016, 2016, pp.~55--62.

\bibitem{doob1953stochastic}
{\sc J.~Doob}, {\em Stochastic Processes}, Probability and Statistics Series,
  Wiley, 1953.

\bibitem{elararphd}
{\sc E.-M. El~Arar}, {\em Stochastic models for the evaluation of numerical
  errors}, PhD thesis, Universit{\'e} Paris-Saclay, 2023.

\bibitem{arar2024probabilistic}
{\sc E.-M. El~Arar, M.~Fasi, S.-I. Filip, and M.~Mikaitis}, {\em Probabilistic
  error analysis of limited-precision stochastic rounding}, 2024,
  \url{https://arxiv.org/abs/2408.03069}.

\bibitem{positive}
{\sc E.-M. El~Arar, D.~Sohier, P.~de~Oliveira~Castro, and E.~Petit}, {\em The
  positive effects of stochastic rounding in numerical algorithms}, in 2022
  IEEE 29th Symposium on Computer Arithmetic (ARITH), 2022, pp.~58--65.

\bibitem{arar2022stochastic}
{\sc E.-M. El~Arar, D.~Sohier, P.~de~Oliveira~Castro, and E.~Petit}, {\em
  Stochastic rounding variance and probabilistic bounds: A new approach}, SIAM
  Journal on Scientific Computing, 45 (2023), pp.~C255--C275.

\bibitem{el2023bounds}
{\sc E.-M. El~Arar, D.~Sohier, P.~de~Oliveira~Castro, and E.~Petit}, {\em
  Bounds on nonlinear errors for variance computation with stochastic
  rounding}, SIAM Journal on Scientific Computing, 46 (2024), pp.~B579--B599.

\bibitem{gupta}
{\sc S.~Gupta, A.~Agrawal, K.~Gopalakrishnan, and P.~Narayanan}, {\em Deep
  learning with limited numerical precision}, in International conference on
  machine learning, PMLR, 2015, pp.~1737--1746.

\bibitem{hallman2022precision}
{\sc E.~Hallman and I.~C.~F. Ipsen}, {\em Precision-aware deterministic and
  probabilistic error bounds for floating point summation}, Numerische
  Mathematik, 155 (2023), pp.~83--119.

\bibitem{ilse}
{\sc I.~C.~F. Ipsen and H.~Zhou}, {\em Probabilistic error analysis for inner
  products}, SIAM Journal on Matrix Analysis and Applications, 41 (2020),
  pp.~1726--1741.

\bibitem{karatsuba1995complexity}
{\sc A.~A. Karatsuba}, {\em The complexity of computations}, Proceedings of the
  Steklov Institute of Mathematics-Interperiodica Translation, 211 (1995),
  pp.~169--183.

\bibitem{azum}
{\sc M.~Mitzenmacher and E.~Upfal}, {\em Probability and Computing: Randomized
  Algorithms and Probabilistic Analysis}, Cambridge University Press, 2005.

\bibitem{muller2018handbook}
{\sc J.-M. Muller, N.~Brisebarre, F.~De~Dinechin, C.-P. Jeannerod, V.~Lefevre,
  G.~Melquiond, N.~Revol, D.~Stehl{\'e}, S.~Torres, et~al.}, {\em Handbook of
  floating-point arithmetic}, vol.~1, Birkh\"{a}user Basel, 2nd~ed., 2018.

\bibitem{xia2023convergence}
{\sc L.~Xia, M.~E. Hochstenbach, and S.~Massei}, {\em On the convergence of the
  gradient descent method with stochastic fixed-point rounding errors under the
  polyak-lojasiewicz inequality}, arXiv preprint arXiv:2301.09511,  (2023).

\bibitem{xia2023influences}
{\sc L.~Xia, S.~Massei, M.~E. Hochstenbach, and B.~Koren}, {\em On the
  influence of stochastic roundoff errors and their bias on the convergence of
  the gradient descent method with low-precision floating-point computation},
  2023, \url{https://arxiv.org/abs/2202.12276}.

\end{thebibliography}

\vspace{1cm}
\appendix
\section{Addition}
\label{sec:ap-addition}

\begin{lemma}
	Let $\mathcal{K}_z=\frac{\abs{x}}{\abs{x+y}}  \mathcal{K}_x+\frac{\abs{y}}{\abs{x+y}} \mathcal{K}_y$. The martingale $(\Psi_i)_{i=0}^{m-1}$ satisfies
		$$ \abs{\Psi_{i} - \Psi_{i-1}} \leq u C_i,
		$$
		where $C_i=\mathcal{K}_z (1+u)^{i-1}$ for all $1\leq i\leq m-1$.
\end{lemma}

\begin{proof}
	For all $0\leq i<m-1$ by definition of $\Psi_i$, we have 
	$$\Psi_{i} - \Psi_{i-1} = \frac{x}{x+y} (\Phi_{\min\{i, k-1\}} - \Phi_{\min\{i-1, k-1\}})+\frac{y}{x+y}(\mathrm{X}_{\min\{i, l-1\}} -\mathrm{X}_{\min\{i-1, l-1\}}).$$
	Then, by induction hypothesis and \mido{\eqref{eq:induction-sum}} we get
	\begin{align*}
		\abs{\Psi_{i} - \Psi_{i-1}} \leq& \abs{\frac{x}{x+y}} \abs{\Phi_{\min\{i, k-1\}} - \Phi_{\min\{i-1, k-1\}}}\\
		&+ \abs{\frac{y}{x+y}} \abs{\mathrm{X}_{\min\{i, l-1\}} -\mathrm{X}_{\min\{i-1, l-1\}} }\\
		\leq& \frac{\abs{x}}{\abs{x+y}} \mathcal{K}_x u(1+u)^{\min\{i-1, k-1\}}  + \frac{\abs{y}}{\abs{x+y}} \mathcal{K}_y u(1+u)^{\min\{i-1, l-1\}} \\
		\leq& \mido{\frac{\abs{x}}{\abs{x+y}} \mathcal{K}_x u(1+u)^{i-1}  + \frac{\abs{y}}{\abs{x+y}} \mathcal{K}_y u(1+u)^{i-1}} \\
		=&  \mido{\mathcal{K}_z u (1+u)^{i-1}.}
	\end{align*}
	Moreover, for $i= m-1$, \mido{ $\Psi_{m-1} = \left(1+ \frac{x}{x+y}\Phi+\frac{y}{x+y}\mathrm{X}\right)(1+\delta) -1$, and $\Psi_{m-2} = \frac{x}{x+y}\Phi+\frac{y}{x+y}\mathrm{X}$. We thus have}
	\begin{align*}
			\abs{\Psi_{m-1} - \Psi_{m-2}} &= \abs{\left(1+ \frac{x}{x+y}\Phi+\frac{y}{x+y}\mathrm{X}\right)\delta}\\
			&= \mido{\abs{\left(\frac{x}{x+y}(\Phi+1)+\frac{y}{x+y}(\mathrm{X}+1)\right)\delta}}\\
			&\leq u \abs{\frac{x}{x+y}(\Phi+1)+\frac{y}{x+y}(\mathrm{X}+1)}.
	\end{align*}
	Since $\Phi_{0}=0$,
	\begin{align*}
		\abs{\Phi+1} &= \abs{\Phi_{k-1}+1} = \abs{1+ \sum_{j=1}^{k-1} (\Phi_{j} - \Phi_{j-1})}\\
		&\leq 1+ \sum_{j=1}^{k-1} \abs{ \Phi_{j} - \Phi_{j-1}}\\
		&\leq 1+ \sum_{j=1}^{k-1} u \mathcal{K}_x (1+u)^{j-1}  &\mido{\text{from \eqref{eq:induction-sum}}}\\
		&= 1+ u\mathcal{K}_x \frac{(1+u)^{k-1} -1}{u}\\
		&= 1+ \mathcal{K}_x (1+u)^{k-1} - \mathcal{K}_x \\
		&\leq \mathcal{K}_x (1+u)^{k-1}  &\mido{\text{since $\mathcal{K}_x \ge 1$}}.
	\end{align*}
	The same method shows that $ \abs{X+1}  \leq \mathcal{K}_y (1+u)^{l-1}$. \mido{Since $m = \max\{k,l\} + 1$,} it follows that
	\begin{align*}
		\abs{\Psi_{m-1} - \Psi_{m-2}}
		&\leq u \left( \frac{\abs{x}}{\abs{x+y}}\mathcal{K}_x (1+u)^{k-1} + \frac{\abs{y}}{\abs{x+y}}\mathcal{K}_y (1+u)^{l-1} \right)\\
		&\leq \mathcal{K}_z u (1+u)^{m-2}.
	\end{align*}
\end{proof}

\begin{theorem}
	For all $0 < \lambda <1$, the computed $\hat{z}$ in Equation~\eqref{eq=z+} satisfies under SR-nearness
	\begin{equation}
		\frac{\abs{\hat z -z}}{\abs{z}} \leq \mathcal{K}_z \sqrt{u\gamma_{2(m-1)}(u)} \sqrt{\ln(2/ \lambda)} = \mathcal{O}(u \sqrt{m}),
	\end{equation}
with probability at least $1-\lambda$, where $\gamma_m(u) = (1+u)^m -1 = \mathcal{O}(mu)$ 
\end{theorem}

\begin{proof}
	Using the Azuma-Hoeffding inequality, we have
	$$ \frac{\abs{\hat z -z}}{\abs{z}}= \abs{\Psi_{m-1}} \leq  \sqrt{\sum_{i=1}^{m-1} u^2 C_i^2} \sqrt{2\ln(2/ \lambda)},
	$$
	with probability at least $1-\lambda$. Moreover,
	\begin{align*}
		\sum_{i=1}^{m-1} u^2 C_i^2 &= u^2  \mathcal{K}_z^2 \sum_{i=1}^{m-1} (1+u)^{2(i-1)}\\
		&= u^2  \mathcal{K}_z^2  \frac{(1+u)^{2(m-1)} -1}{u^2 +2u}\\
		&\leq u  \mathcal{K}_z^2  \frac{\gamma_{2(m-1)}(u)}{2}.
	\end{align*}
Finally, we get
\begin{align*}
	\frac{\abs{\hat z -z}}{\abs{z}} \leq \mathcal{K}_z \sqrt{u\gamma_{2(m-1)}(u)} \sqrt{\ln(2/ \lambda)},
\end{align*}
with probability at least $1-\lambda$.
\end{proof}

\section{Multiplication}
\label{sec:ap-multiplication}

\begin{lemma}
	Let $\mathcal{K}_z= \mathcal{K}_x \mathcal{K}_y$. The martingale $(\Psi_i)_{i=0}^{m-1}$ satisfies
	$$ \abs{\Psi_{i} - \Psi_{i-1}} \leq u C_i,
	$$
	where $C_i=\mathcal{K}_z (1+u)^{i-1}$ for all $1\leq i\leq m-1$.
\end{lemma}

\begin{proof}
	For all $1\leq i \leq k-1$, $\abs{\Psi_{i} - \Psi_{i-1}}= \abs{\Phi_{i} - \Phi_{i-1}} \leq \mathcal{K}_x u (1+u)^{i-1}$. Moreover, for all $k \leq i \leq m-2$,
	\begin{align*}
		\abs{\Psi_{i} - \Psi_{i-1}} &= \abs{(1+\Phi)(1+\mathrm{X}_i) - (1+\Phi)(1+\mathrm{X}_{i-1})}\\
		&= \abs{(1+\Phi)(\mathrm{X}_i - \mathrm{X}_{i-1})}\\
		&\leq   \abs{1+\Phi} \mathcal{K}_y u (1+u)^{i-k}.
	\end{align*}
	As for the summation case, $\abs{1+\Phi} \leq \mathcal{K}_x (1+u)^{k-1}$. Then, for all $k \leq i \leq m-2$,
	$$ 	\abs{\Psi_{i} - \Psi_{i-1}} \leq \mathcal{K}_x (1+u)^{k-1}\mathcal{K}_y u (1+u)^{i-k} = \mathcal{K}_z u (1+u)^{i-1}.
	$$
	Finally, \mido{for $i=m-1$,} we obtain
	\begin{align*}
		\abs{\Psi_{m-1} - \Psi_{m-2}} &= \abs{(1+\Phi)(1+\mathrm{X})\delta}\\
		&\leq u \mathcal{K}_x (1+u)^{k-1} \mathcal{K}_y (1+u)^{m-k-1}\\
		&\leq  u \mathcal{K}_z (1+u)^{m-2}. 
	\end{align*}
\end{proof}

\begin{theorem}
	For all $0 < \lambda <1$, the computed $\hat{z}$ in Equation~\eqref{eq:z*} satisfies under SR-nearness
	\begin{equation}
		\frac{\abs{\hat z -z}}{\abs{z}} \leq \mathcal{K}_z \sqrt{u\gamma_{2(m-1)}(u)} \sqrt{\ln(2/ \lambda)} \mido{= \mathcal{O}(u \sqrt{m})},
	\end{equation}
	with probability at least $1-\lambda$.
\end{theorem}

\begin{proof}
	Using the Azuma-Hoeffding inequality, we have
	$$ \frac{\abs{\hat z -z}}{\abs{z}}= \abs{\Psi_{m-1}} \leq  \sqrt{\sum_{i=1}^{m-1} u^2 C_i^2} \sqrt{2\ln(2/ \lambda)},
	$$
	with probability at least $1-\lambda$. Moreover,
	\begin{align*}
		\sum_{i=1}^{m-1} u^2 C_i^2 &= u^2  \mathcal{K}_z^2 \sum_{i=1}^{m-1} (1+u)^{2(i-1)}\\
		&= u^2  \mathcal{K}_z^2  \frac{(1+u)^{2(m-1)} -1}{u^2 +2u}\\
		&\leq u  \mathcal{K}_z^2  \frac{\gamma_{2(m-1)}(u)}{2}.
	\end{align*}
	Finally, we get
	\begin{align*}
		\frac{\abs{\hat z -z}}{\abs{z}} \leq \mathcal{K}_z \sqrt{u\gamma_{2(m-1)}(u)} \sqrt{\ln(2/ \lambda)},
	\end{align*}
	with probability at least $1-\lambda$.
\end{proof}

\section{Karatsuba polynomial multiplication}
\label{sec:ap-karatsuba}

\begin{theorem}
  For $i \in \llbracket 0, d \rrbracket$, the length of the error martingale in the computation of coefficient $r_i$ is
  \begin{equation*}
    m(i, d) = 1 + 3 \lfloor\log_2 \min\{i+1, d-i+1\}\rfloor + m_A + m_B.
  \end{equation*}
\end{theorem}
\begin{proof}
  By induction on $n$.

  For $n=0$, $d=0$ and $R = r_0 = a_0.b_0$. Because the DAG is composed of a single multiplication, the length of the error martingale is $1 + m_A + m_B$.
  We verify that this is the value of $m(0,0)$.

  For $n \ge 1$, $d=2^{n+1}-2$, we consider $R = \mido{K(A,B)}$ where $A$ and $B$ are of degree $2^n-1$.
  Let us first compute the martingale lengths for the partial products $P_0$, $P_1$ and $P_2$:

  \begin{itemize}
    \item $P_0 = K(A_l, B_l)$, where $A_l$ and $B_l$ are of degree $2^{n-1}-1$. By induction hypothesis, the length of the error martingale is $m_{P_0}(i) = m(i, 2^{n}-2)$.
    \item $P_2 = K(A_h, B_h)$, similarly the length of the error martingale is $m_{P_2}(i) = m(i, 2^{n}-2)$.
    \item $P_1 = K(A_h - A_l, B_l - B_h)$.
          Each coefficient in $A_h-A_l$ is computed by subtracting two coefficients from $A_h$ and $A_l$, therefore its maximum martingale length is $1+\max\{m_A, m_A\} = 1 + m_A$.
          The same reasoning applies to $B_h - B_l$, which has a maximum martingale length $1+m_B$.
          By induction hypothesis, the length of the error martingale of $P_1$ is $m_{P_1}(i) = m(i, 2^{n}-2) + 2$.
          As shown, the additional $2$ comes from the inner subtractions.
  \end{itemize}

  Let us now consider $R = P_2.X^{2^{n}} + (P_0 + P_1 + P_2).X^{2^{n-1}} + P_0$.
  Figure~\ref{fig:shifts} represents the shifted partial products in the
  computation of $R$ : $P_0$ is not shifted, $P_2$ is shifted $2^{n}$
  positions left. $P_1$, $P_2$ and $P_0$ are shifted $2^{n-1}$ positions left.

  \begin{figure}
  \hspace{-.25cm}
  \scalebox{0.96}{
    \begin{tikzpicture}[
        start chain = going right,
        node distance = 0pt,
        tops/.style={draw, minimum width=1.5cm, minimum height=1cm,
            outer sep=0pt, on chain},
        bots1/.style={draw, minimum height=1cm, minimum width=6cm, outer sep=0pt},
        bdash/.style={draw, dashed, minimum height=1cm, minimum width=1.5cm, outer sep=0pt},
      ]
      \node [tops] (n1)  at (0,0) {\small$d$};
      \node [tops,minimum width=2.25cm] (n2) {$\cdots$};
      \node [tops] (n3) {\small$d-2^{n-1}$};
      \node [tops,minimum width=.75cm] (n4) {$\cdots$};
      \node [tops] (n5) {\small$d/2$};
      \node [tops,minimum width=.75cm] (n6) {$\cdots$};
      \node [tops] (n7) {\small$2^{n-1}$};
      \node [tops,minimum width=2.25cm] (n8) {$\cdots$};
      \node [tops] (n9) {\small$0$};

      \node   [bots1]  at (2.25cm,-1.5 cm) {};
      \node   [bdash] at (2.25cm, -1.5cm) {$P_2$};
      \node   [bots1]  at (9.75cm,-1.5 cm) {};
      \node   [bdash] at (9.75cm, -1.5cm) {$P_0$};
      \node   [bots1]  at (6cm,-3 cm) {};
      \node   [bots1]  at (6cm,-4.5 cm) {};
      \node   [bots1]  at (6cm,-6 cm) {};
      \node   [bdash]  at (6cm,-3 cm) {$P_0$};
      \node   [bdash]  at (6cm,-4.5 cm) {$P_1$};
      \node   [bdash]  at (6cm,-6 cm) {$P_2$};

      \draw [decorate,decoration={brace,amplitude=10pt}]
      (-.75cm, 1cm) -- (3cm,1cm) node[black,midway,above=8pt] {Case (c)};
      \draw [decorate,decoration={brace,amplitude=10pt}]
      (9cm, 1cm) -- (12.75cm,1cm) node[black,midway,above=8pt] {Case (a)};
      \draw [decorate,decoration={brace,amplitude=10pt}]
      (3cm, 1cm) -- (9cm,1cm) node[black,midway,above=8pt] {Case (b)};
    \end{tikzpicture}
    }
    \caption{Shifted partial products in $R = K(A,B)$ with $d=2^{n+1}-2$. The dashed cells represent the central coefficient in each polynomial. \label{fig:shifts}}
  \end{figure}

  Let us treat separately cases (a), (b), and (c).

  \paragraph{Case (a):} For $i \in \llbracket 0, 2^{n-1}-1 \rrbracket$, $r_i$ corresponds to the $i$-th coefficient in $P_0$.
  Therefore the martingale length for $r_i$ is given by $m_{P_0}(i)$ and
  \begin{equation*}
    m_{P_0}(i) = m(i, 2^{n}-2) = 1 + 3 \lfloor\log_2 (i+1)\rfloor + m_A + m_B = m(i, d)
  \end{equation*}

  \paragraph{Case (c):} For $i \in \llbracket d-2^{n-1}+1, d \rrbracket$, $r_i$ corresponds to the $i-2^n$ coefficient in $P_2$.
  Therefore the martingale length for $r_i$ is given by $m_{P_2}(i-2^n)$ and
  \begin{align*}
    m_{P_2}(i-2^n) = m(i-2^n, 2^{n}-2) & = 1 + 3 \lfloor\log_2 \left((2^n-2) - (i-2^n) + 1\right)\rfloor + m_A + m_B \\
                                   & = 1 + 3 \lfloor\log_2 \left((2^{n+1}-2) -i + 1\right)\rfloor + m_A + m_B    \\
                                   & = m(i, d)
  \end{align*}
  \paragraph{Case (b):} This case is the most interesting one, because each
  coefficient $r_i$ results of the sum of at most four coefficients from $P_0$, $P_2$,
  $P_0.X^{2^{n-1}}$, $P_1.X^{2^{n-1}}$, and $P_2.X^{2^{n-1}}$.

  To minimize the martingale length, we will use the summing order from the DAG in Figure~\ref{fig:sumd}.
  Note that even if the figure depicts a tree, it corresponds to a DAG since some leaves share coefficients (e.g. $P_2$ and $P_2.X^{2^{n-1}}$).
  Let us note $m_b(i)$ the martingale length of the DAG.

  \begin{figure}
    \centering
    \begin{tikzpicture}[level distance=1.5cm,
        level 1/.style={sibling distance=3cm},
        level 2/.style={sibling distance=4cm},
        level 3/.style={sibling distance=2cm},
        edge from parent/.style={<-,draw},
        >=latex
      ]
      \node {$+$}
      child {node {$P_1.X^{2^{n-1}}$}}
      child {node {$+$}
      child {node {$+$}
      child {node {$P_2.X^{2^{n}}$}}
      child {node {$P_2.X^{2^{n-1}}$}}
      }
      child {node {$+$}
      child {node {$P_0.X^{2^{n-1}}$}}
      child {node {$P_0$}}
      }
      };
    \end{tikzpicture}
    \caption{Summing DAG for case (b) \label{fig:sumd}.}
  \end{figure}

  For each node $z \leftarrow x + y$ in the tree, $m_z = 1+ \max\{m_x,m_y\}$.
  Therefore, the martingale length for $r_i$ is given by
  \begin{align*}
    m_b(i) & = 1 + \max\Big\{  m_{P_1}(i-2^{n-1}), 1 + \max\{                             \\
           & \quad\quad\quad 1 + \max\{m_{P_2}(i-2^{n}), m_{P_2}(i-2^{n-1})\},        \\
           & \quad\quad\quad 1 + \max\{m_{P_0}(i-2^{n-1}), m_{P_0}(i)\}\}\Big\}             \\
           & = 1 + \max\Big\{  2 + m(i-2^{n-1}, 2^n-2), 1 + \max\{                        \\
           & \quad\quad\quad 1 + \max\{m(i-2^{n}, 2^{n}-2), m(i-2^{n-1}, 2^{n}-2)\},  \\
           & \quad\quad\quad 1 + \max\{m(i-2^{n-1}, 2^{n}-2), m(i, 2^{n}-2)\}\}\Big\}       \\
           & = 1 + \max\Big\{  2 + m(i-2^{n-1}, 2^n-2), \\
           & \quad\quad\quad 2 + \max\{m(i-2^{n}, 2^{n}-2), m(i-2^{n-1}, 2^{n}-2), m(i, 2^{n}-2)\}\Big\}  \\
           & = 3 + \max\Big\{ m(i-2^{n}, 2^{n}-2), m(i-2^{n-1}, 2^{n}-2), m(i, 2^{n}-2)\Big\}
  \end{align*}
  We note that to achieve a minimal martingale length, it is important to have
  $P_1$, which has an additional martingale length of 2, as a direct child of the root node.

  Now let us compute $m_b(i)$:
  \begin{itemize}
    \item For $i = d/2 = 2^{n}-1$, the maximum is reached for $P_0$, $P_1$, and $P_2$ due to \emph{Property 2},
          \begin{align*}
            m_b(i) & =  3 + m(2^{n}-1-2^{n-1}, 2^{n}-2) = 3 + m(2^{n-1}-1, 2^{n}-2)      \\
                   & =         3 + 3(n-1) + 1 + m_A + m_B = 1 + 3n + m_A + m_B = m(i, d)
          \end{align*}
    \item For $i \in \llbracket 2^{n-1}, d/2 \llbracket$, we subdivide the interval in two.
          \begin{itemize}
            \item For $i \in \llbracket 2^{n-1}, 2^{n-1}+2^{n-2} \llbracket$, we are close to the central element of $P_0$, therefore due to \emph{Property 3}
                  \begin{align*}
                    m_b(i) & =  3 + 3(n-2) + 1 + m_A + m_B = 3(n-1) + 1 + m_A + m_b     \\
                           & =  m(i, d)                 \end{align*}
            \item For $i \in \llbracket 2^{n-1}+2^{n-2}, d/2 \llbracket$, we are close to the central element of $P_1$ (or $-P_2, -P_0$), therefore due to \emph{Property 3} we conclude as before.
          \end{itemize}
    \item For $i \in \rrbracket d/2, d-2^{n-1} \rrbracket$, we apply a similar proof scheme by subdividing the interval and applying \emph{Property 3}.
  \end{itemize}
\end{proof}

\begin{remark}{Computation of $\mathcal{K}_{d/2}$ for polynomials with constant positive coefficients.} For $n \geq 0$, we define

$$
C(A, B) = \mathcal{K}_{d/2} \cdot |r_{d/2}|,
$$

where $\mathcal{K}_{d/2}$ is the condition number of the central coefficient of the product $A \times B$, and $r_{d/2}$ denotes the central coefficient of the product.

Assume that both \(A\) and \(B\) are polynomials whose coefficients are all equal to a constant \(a > 0\). 
We claim that $C_n(a) = a^2 \cdot 6^n.$

\begin{itemize}
\item
  \textbf{Base case ($n = 0$, $d = 0$)}:
  $$
  C_0 = C(A, B) = |a_0 \cdot b_0| = a^2 \cdot 6^0.
  $$
\item
  \textbf{Recursive case ($n \geq 1$)}:
  $$
  C_n = C(A, B) = |C_{P_0}| + |C_{P_2}| + |C_{P_1}|,
  $$

  where
  $$ 
  C_{P_0} = C(A_l, B_l) = C_{n-1}(a), \quad C_{P_2} = C(A_h, B_h) = C_{n-1}(a),
  $$

  since $A_l, B_l, A_h, B_h$ all have coefficients equal to $a$. Furthermore,

  $$C_{P_1} = C(A_l + A_h, B_l + B_h) = C_{n-1}(2a),$$

  as both $A_l + A_h$ and $B_l + B_h$ are polynomials with
  constant coefficients equal to $2a$. Note that since we are computing the numerator of the condition number, subtractions have been replaced by additions in the cross product.  Hence, we obtain the recurrence:
  $$
  C_n(a) = 2 \cdot C_{n-1}(a) + C_{n-1}(2a),
  $$
  
  which evaluates to $ C_n(a) = 2a^2 \cdot 6^{n-1} + 4a^2 \cdot 6^{n-1} = a^2 \cdot 6^n.$
\end{itemize}

Moreover, by induction $r_{d/2} = a^2 \cdot 2^n$.
It follows that $\mathcal{K}$ grows with $d$:
$$
\mathcal{K}_{d/2} = \frac{C_n(a)}{|r_{d/2}|} = \frac{a^2 \cdot 6^n}{a^2 \cdot 2^n} = 3^n = (d + 2)^{\log_2 3}.
$$

\end{remark}
\end{document}